\documentclass[twocolumn,superscriptaddress,
showpacs,preprintnumbers,amsmath,amssymb]{revtex4}

\usepackage{graphicx}
\usepackage{amssymb,amsmath,amsthm}
\newtheorem{Thm}{Theorem}

\newtheorem{Lem}[Thm]{Lemma}
\newtheorem{Prop}[Thm]{Proposition}
\theoremstyle{definition}

\newcommand{\bra}[1]{{\left\langle #1 \right|}}
\newcommand{\ket}[1]{{\left| #1 \right\rangle}}

\begin{document}

\title{Any multipartite entangled state violating Mermin-Klyshko inequality
can be distilled
for almost all bipartite splits}

\author{Soojoon Lee}
\affiliation{
 Department of Mathematics and Research Institute for Basic Sciences,
 Kyung Hee University, Seoul 130-701, Korea
}
\author{Jinhyoung Lee}
\affiliation{
 Department of Physics,
 Hanyang University,
 Seoul 133-791, Korea
}
\author{Jaewan Kim}
\affiliation{
 School of Computational Sciences,
 Korea Institute for Advanced Study,
 Seoul 130-722, Korea
}
\date{\today}

\begin{abstract}
We study the explicit relation between violation of Bell inequalities and
bipartite distillability of multi-qubit states.
It has been shown that even though for $N\ge 8$
there exist $N$-qubit bound entangled states
which violates a Bell inequality~[
Phys. Rev. Lett. {\bf 87}, 230402 (2001)],
for all the states violating the inequality
there exists at least one splitting of the parties into two groups
such that pure-state entanglement can be distilled~[
Phys. Rev. Lett. {\bf 88}, 027901 (2002)].
We here prove that
for all $N$-qubit states violating the inequality
the number of distillable bipartite splits increases
exponentially with $N$,
and hence the probability that a randomly chosen bipartite split is distillable
approaches one exponentially with $N$, as $N$ tends to infinity.
We also show that there exists at least one $N$-qubit bound entangled state violating the inequality
if and only if $N\ge 6$.
\end{abstract}

\maketitle

Entanglement has been considered as a key ingredient for quantum information science,
and has brought a lot of its useful applications such as quantum key distribution and teleportation.
Nevertheless, there still exist open problems related to entanglement,
in particular, multipartite entanglement.

It is known that
entanglement can be divided into two kinds of entanglement.
One is called the {\em distillable} entanglement,
from which some pure entanglement can be extracted by
local quantum operations and classical communication,
and the other is called the {\em bound} entanglement,
which is not distillable.
Since only pure entanglement is directly useful for quantum information processing,
the bound entanglement seems to be useless.
However, it has been recently shown that
any bound entangled (BE) states are useful in quantum teleportation~\cite{HHH,Masanes1},
all multipartite pure entangled states are interconvertible
by stochastic local operations and classical communication
with the assistance of BE states~\cite{Ishizaka},
and there are several classes of BE states
with a positive key rate in quantum key distribution~\cite{HHHO1,HHHO2,HPHH,CCKKL,HA}.
Thus, it is necessary to analyze BE states more profoundly.

If one of the two most significant features related to entanglement is distillability,
then the other is {\em nonlocality},
which can be described as a physical property to explain
that quantum correlation is quite different from classical correlations.
Nonlocality can be seen from violation of some conditions, called Bell inequalities,
that are satisfied by any local variable theory,
and it is a well-known fact that any bipartite or multipartite pure state violates a Bell inequality
if and only if the state is entangled~\cite{Gisin91,PR92}.
However, for mixed states,
there does not exist such a simple relation between nonlocality and entanglement.
Since Werner~\cite{Werner} found the existence of entangled mixed states
described by a local hidden variable model,
it has been known that some of these states can violate Bell inequalities
after appropriately preprocessing the state~\cite{Popescu,Gisin96}.

There is a simple relation between nonlocality and distillability in fewer-qubit systems:
If any two-qubit~\cite{
HHH97} or three-qubit~\cite{LJK} (pure or mixed) state
violates a specific form of the Bell inequality
then it is distillable.
However, D\"{u}r~\cite{Dur} has shown that for $N\ge 8$
there exist $N$-qubit BE states
which violate a Bell inequality.
This result seems to show that nonlocality does not directly
imply distillability in multipartite cases,
even though it has been recently shown that
asymptotic violation of a Bell inequality is equivalent to distillability
in any multipartite quantum system~\cite{Masanes}.

But, Ac\'{\i}n~\cite{Acin} has demonstrated that
for all the states violating the inequality
there exists at least one splitting of the parties into two groups
such that pure-state entanglement can be distilled,
and has more analyzed
the relation of nonlocality to bipartite distillability
in his subsequent works~\cite{Acin_etal}.
This does not only imply that
there still exists a relation between nonlocality and
distillability for a certain bipartite split,
but also tells us that
it is possible to make
two-party quantum communications with respect to the bipartite split
secure against eavesdropping.
Then some questions naturally arise such as
which bipartite split is distillable
and how many splits are possible to be distillable if the Bell inequality is violated.

Assume that a multipartite entangled state violates the Bell inequality.
If it could be distilled for almost all bipartite splits,
then it would be possible
for almost all two-party quantum communications over the multipartite state
to be secure against eavesdropping,
regardless of how it is divided into two parties.
Thus, it would be important to answer the questions
in quantum communication theory as well as in entanglement theory.

In this paper, we show that
if any $N$-qubit state violates the inequality then
there exist much more than one distillable bipartite splits,
to be exact,
at least $\lfloor 2^{N-1} - 2^{(N-1)/2}+1\rfloor$ distillable bipartite splits.
Hence, the distillation probability that a randomly chosen bipartite split is distillable
approaches one exponentially with $N$ as $N$ tends to infinity.
This means that
if a given $N$-qubit state violates the Bell inequality for sufficiently large $N$
then almost all bipartite splits are distillable.
Furthermore, this result provides us with
the following necessary and sufficient condition for
the existence of $N$-qubit BE states violating the inequality:
At least one $N$-qubit BE state violates the inequality
if and only if $N\ge 6$.

Since it has been already known that
there exists a four-qubit BE state, the so-called Smolin state~\cite{Smolin},
violating some other Bell inequality~\cite{AH},
our condition does not seem to be very strong.
However, because our proof is based on the first main result
counting distillable bipartite splits,
this justifies some significance of considering the counting problem.

In order to introduce our main results,
we first consider the family of $N$-qubit states $\rho_N$ presented in~\cite{DCT,DC},
\begin{eqnarray}
\rho_N &=& \sum_{\sigma=\pm}\lambda_0^\sigma\ket{\Psi_0^\sigma}\bra{\Psi_0^\sigma}
\nonumber\\
&&+\sum_{j=1}^{2^{N-1}-1}\lambda_j
\left(\ket{\Psi_j^+}\bra{\Psi_j^+}+\ket{\Psi_j^-}\bra{\Psi_j^-}\right),
\label{eq:rho_N}
\end{eqnarray}
where
\begin{equation}
\ket{\Psi_j^\pm}=\frac{1}{\sqrt{2}}
\left(\ket{j}\ket{0}\pm\ket{2^{N-1}-j-1}\ket{1}\right),
\label{eq:Phi_j}
\end{equation}
and $\lambda_0^+ +\lambda_0^- +2\sum_j \lambda_j =1$.
We remark that any arbitrary $N$-qubit state can be depolarized to a state in this family,
and hence this family can be useful to find sufficient conditions
for nonseparability and distillability
in $N$-qubit systems~\cite{DCT}.
Thus, this family may be regarded
as a generalization of Werner states to multiqubit systems.

We prove our first main result
in the following way:
(i) We assume that any $N$-qubit state $\rho$ violates a specific form of Bell inequality.
(ii) By some appropriate depolarizing process,
the state $\rho$ can be transformed into $\rho_N$,
which also violates the same inequality.
(iii) We show that the state $\rho_N$ violating the inequality has
at least $\lfloor 2^{N-1} - 2^{(N-1)/2}+1\rfloor$ distillable bipartite splits.
(iv) We conclude that the state $\rho$ also has
at least $\lfloor 2^{N-1} - 2^{(N-1)/2}+1\rfloor$ distillable bipartite splits.
In order to prove the main result,
we need the following proposition and lemma.

For each ($N-1$)-bit string $j=j_1j_2\cdots j_{N-1}$,
let $P_j$ be the bipartite split such that
$j_i=0$ if and only if party $i$ belongs to the same set as the last party.
Then the following proposition about bipartite distillability of the states $\rho_N$
has been known by D\"{u}r and Cirac~\cite{DC}.
\begin{Prop}\label{Prop:DC}
$\rho_N$ is distillable for the bipartite split $P_j$
if and only if  $2\lambda_j<\Delta\equiv\lambda_0^+ -\lambda_0^-$.
\end{Prop}
We note that the quantity $\Delta$ in Proposition~\ref{Prop:DC}
plays an important role in not only bipartite distillability
but also a certain form of Bell inequality,
which we will crucially use in this paper.

From Proposition~\ref{Prop:DC},
we can obtain the following key lemma for our first main result.
\begin{Lem}\label{Lem:main}
If
\begin{equation}
\Delta > \frac{1}{2^{(N-1)/2}}
\label{eq:violation}
\end{equation}
then there exist
at least $\lfloor 2^{N-1} - 2^{(N-1)/2}+1\rfloor$ distillable bipartite splits
in $\rho_N$.
\end{Lem}
\begin{proof}
Let $m$ be the number of distillable bipartite splits,
$P_{j_1}, P_{j_2}, \ldots , P_{j_m}$.
Suppose that $m \le 2^{N-1} - 2^{(N-1)/2}$.
Then we readily obtain the following inequality: 
\begin{eqnarray}
1-\Delta 
&\ge& 
2\sum_{j=1}^{2^{N-1}-1} \lambda_j \nonumber \\
&=& 2(\lambda_{j_1}+\lambda_{j_2}+\cdots+ \lambda_{j_m})
+2\sum_{j\notin \{j_1,\ldots,j_m\}} \lambda_j \nonumber \\
&\ge& 2(\lambda_{j_1}+\lambda_{j_2}+\cdots+ \lambda_{j_m}) 
+(2^{N-1}-1-m)\Delta.\nonumber \\
\label{eq:ineqs01}
\end{eqnarray}
It follows that
\begin{eqnarray}
1&\ge&  2(\lambda_{j_1}+\lambda_{j_2}+\cdots+ \lambda_{j_m})
+(2^{N-1}-m)\Delta\nonumber \\
&>&  2(\lambda_{j_1}+\lambda_{j_2}+\cdots+ \lambda_{j_m})
+(2^{N-1}-m)/2^{(N-1)/2} \nonumber \\
&\ge& 2(\lambda_{j_1}+\lambda_{j_2}+\cdots+ \lambda_{j_m}) +1.
\label{eq:ineqs02}
\end{eqnarray}
The inequality~(\ref{eq:ineqs02})
leads to a contradiction. 
Therefore, we can conclude that $m> 2^{N-1} - 2^{(N-1)/2}$.
\end{proof}

The Bell inequality that D\"{u}r and Ac\'{\i}n have considered
is called the Mermin-Klyshko (MK) inequality~\cite{Mermin,BK},
which generalizes the Clauser-Horne-Shimony-Holt inequality~\cite{CHSH} into $N$-qubit cases.
Let $\mathcal{B}_N$ be the Bell operator defined recursively as
\begin{equation}
\mathcal{B}_i=\frac{1}{2}\left[
\mathcal{B}_{i-1}\otimes\left(\sigma_{\hat{n}_i}+\sigma_{\hat{n}'_i}\right)
+\mathcal{B}'_{i-1}\otimes\left(\sigma_{\hat{n}_i}-\sigma_{\hat{n}'_i}\right)
\right],
\label{eq:Bell_operator}
\end{equation}
where $\sigma_{\hat{n}_i}=\hat{n}_i\cdot \sigma$ and $\sigma_{\hat{n}'_i}=\hat{n}'_i\cdot \sigma$
are the two dichotomic observables measured on each particle $i$,
$\mathcal{B}'_i$ is obtained from $\mathcal{B}_i$ by exchanging all the $\hat{n}_i$ and $\hat{n}'_i$,
and $\mathcal{B}_1=\sigma_{\hat{n}_1}$.
Then the MK inequality is as follows:
\begin{equation}
\left|\mathrm{tr}\left(\mathcal{B}_N\rho\right)\right|\le 1.
\label{eq:MK_ineq0}
\end{equation}

Choosing the same measurement directions in all $N$ locations,
$\sigma_{\hat{n}_i}=\sigma_x$ and $\sigma_{\hat{n}'_i}=\sigma_y$ for all $i$,
after local phase redefinition~\cite{Acin},
$\mathcal{B}_N$ can be written as
\begin{equation}
\mathcal{B}_{N}=2^{(N-1)/2}
\left(\ket{\Psi_0^+}\bra{\Psi_0^+}-\ket{\Psi_0^-}\bra{\Psi_0^-}\right).
\label{eq:B_N}
\end{equation}
We note that, by the depolarizing process in~\cite{DCT},
any $N$-qubit state $\rho$ can be transformed into
one in the family of $\rho_N$
with $\lambda_0^\pm =\bra{\Psi_0^\pm}\rho_N\ket{\Psi_0^\pm}=\bra{\Psi_0^\pm}\rho\ket{\Psi_0^\pm}$
and $2\lambda_j= \bra{\Psi_j^+}\rho_N\ket{\Psi_j^+}+\bra{\Psi_j^-}\rho_N\ket{\Psi_j^-}
=\bra{\Psi_j^+}\rho\ket{\Psi_j^+}+\bra{\Psi_j^-}\rho\ket{\Psi_j^-}$.
Thus,  for the Bell operator $\mathcal{B}_{N}$ in Eq.~(\ref{eq:B_N}),
we obtain the following equalities:
\begin{eqnarray}
2^{-(N-1)/2}\mathrm{tr}\left(\mathcal{B}_{N}\rho\right)
&=&
\bra{\Psi_0^+}\rho\ket{\Psi_0^+}-\bra{\Psi_0^-}\rho\ket{\Psi_0^-}
\nonumber\\
&=&
\bra{\Psi_0^+}\rho_{N}\ket{\Psi_0^+}-\bra{\Psi_0^-}\rho_{N}\ket{\Psi_0^-}
\nonumber\\
&=&
\lambda_0^+ - \lambda_0^- =\Delta,
\label{eq:lambda_Delta}
\end{eqnarray}
and hence we have the following theorem by Lemma~\ref{Lem:main}.
\begin{Thm}\label{Thm:main}
For all the $N$-qubit states $\rho$ violating the MK inequality
with respect to the Bell operator~(\ref{eq:B_N}),
there exist
at least $\lfloor 2^{N-1} - 2^{(N-1)/2}+1\rfloor$ distillable bipartite splits.
\end{Thm}

Let $P(N)$ be the probability that
a randomly chosen bipartite split on an $N$-qubit state is distillable,
when it violates the MK inequality
with respect to the Bell operator~(\ref{eq:B_N}).
Then it follows from Theorem~\ref{Thm:main} that
\begin{eqnarray}
P(N)
&\ge& \frac{2^{N-1} - 2^{(N-1)/2}}{2^{N-1}-1} \nonumber \\
&=& 1- \frac{1}{2^{(N-1)/2}+1}.
\label{eq:P_N}
\end{eqnarray}
This implies that
$P(N)$ approaches one exponentially with $N$ as $N$ tends to infinity
as seen in FIG.~\ref{Fig:Pr_N}.

\begin{figure}
\includegraphics[width=\linewidth]{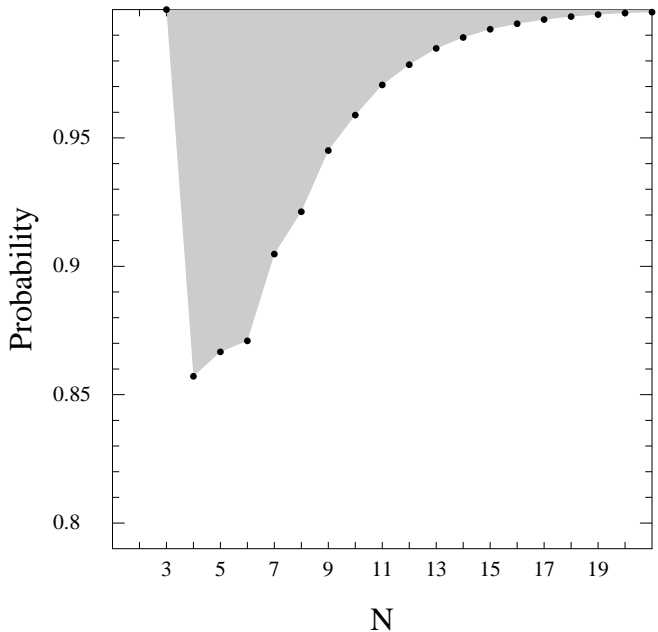}
\caption{\label{Fig:Pr_N}
The distillation probability $P(N)$ that
a randomly chosen bipartite split on an $N$-qubit state is distillable,
when it violates the MK inequality
with respect to the Bell operator~(\ref{eq:B_N}).
}
\end{figure}

Interestingly, Theorem~\ref{Thm:main} provides us with
a necessary and sufficient condition for
the existence of $N$-qubit BE states violating the MK inequality
with respect to the Bell operator~(\ref{eq:B_N}).
In order to show the condition,
we begin with reminding the following proposition
about a relation between distillability and negative partial transposition (NPT),
which has been shown by D\"{u}r and Cirac~\cite{DCT}.
\begin{Prop}\label{Prop:DCT}
A maximally entangled pair between particles $i$ and $j$
can be distilled from $\rho_N$ if and only if
all possible bipartite splits of $\rho_N$
where the particles $i$ and $j$ belong to different parties,
have NPT.
\end{Prop}

By Theorem~\ref{Thm:main} and Proposition~\ref{Prop:DCT},
we can prove the following theorem.
\begin{Thm}\label{Thm:main2}
There exists at least one $N$-qubit BE state violating the MK inequality
with respect to the Bell operator~(\ref{eq:B_N})
if and only if $N\ge 6$.
\end{Thm}
\begin{proof}
We note that the number of total bipartite splits is $2^{N-1}-1$,
and that the number of all distillable bipartite splits is
at least $\lfloor 2^{N-1} - 2^{(N-1)/2}+1\rfloor$ by Theorem~\ref{Thm:main}.

We first assume that $N\le 5$, that is, $N=3$, $N=4$, or $N=5$.

(Case 1) $N=3$;
It follows from Theorem~\ref{Thm:main} that all bipartite splits are distillable,
and so have NPT.
By Proposition~\ref{Prop:DCT},
a maximally entangled state can be distilled
between any particles $i$ and $j$.

(Case 2) $N=4$;
Since $\lfloor 2^3 - 2^{3/2}+1\rfloor=6$ and $2^3-1=7$,
we obtain that all bipartite splits are distillable or
there is only one non-distillable bipartite split.
Hence, there is at least one pair $i$ and $j$
such that all bipartite splits whose two different parties contain the particles $i$ and $j$ respectively
are distillable.
As in the Case 1, since they have NPT,
a maximally entangled pair can be distilled between the particles $i$ and $j$.

(Case 3) $N=5$;
Since $\lfloor 2^4 - 2^2+1\rfloor=13$ and $2^4-1=15$,
we obtain that all bipartite splits are distillable, or
there exist at most two non-distillable bipartite splits.
Hence, there is at least one pair $i$ and $j$
between which a maximally entangled pair can be distilled
by the same reason as the Case 2.

Conversely, if $N\ge 6$ then
there exists an $N$-qubit BE state violating the MK inequality
as follows:
Take $\lambda_0^+=1/(N-1)$, $\lambda_0^-=0$, and
$\lambda_j=1/2(N-1)$ if $j=3, 6, \ldots, 3\cdot 2^{N-3}$
and $\lambda_j=0$ otherwise.
Under these conditions, the state $\rho_N$ becomes,

\begin{eqnarray}
\varrho_N&=&\frac{1}{N-1}\ket{\Psi_0^+}\bra{\Psi_0^+}
\nonumber \\
&&+\frac{1}{2(N-1)}\sum_{j\in J_N}
\left(\ket{\Psi_j^+}\bra{\Psi_j^+}+\ket{\Psi_j^-}\bra{\Psi_j^-}\right),
\nonumber \\
\label{eq:BErho_N}
\end{eqnarray}
where $J_N=\{3, 6, \ldots, 3\cdot 2^{N-3}\}$.
Then since $N-1 < 2^{(N-1)/2}$ if $N\ge 6$,
the state $\varrho_N$ violates the MK inequality
with respect to the Bell operator~(\ref{eq:B_N}).

\begin{figure}
\includegraphics[width=0.7\linewidth]{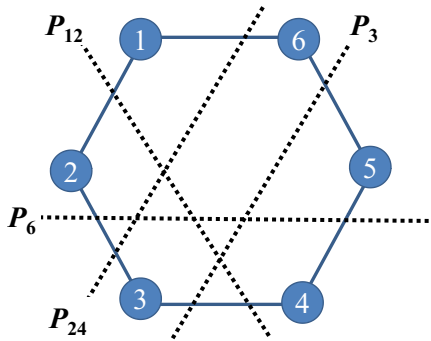}
\caption{\label{Fig:splits}
Undistillable bipartite splits $P_j$ of $\varrho_N$ in (\ref{eq:BErho_N}) when $N=6$.
}
\end{figure}

Furthermore, since $\Delta=2\lambda_j$ if $j\in J_N$,
by Proposition~\ref{Prop:DC},
the state $\varrho_N$ is not distillable
for the bipartite splits $P_j$ for $j\in J_N$.

As seen in FIG.~\ref{Fig:splits},
if two different particles $k$ and $k'$ in the state $\varrho_N$ are given
then $P_{3\cdot 2^{N-1-k}}$ or $P_{3\cdot 2^{N-2-k}}$ is a bipartite split
where the two particles belong to different parties,
and neither $P_{3\cdot 2^{N-1-k}}$ nor $P_{3\cdot 2^{N-2-k}}$
is bipartite distillable,
and hence a maximally entangled state between the particles $k$ and $k'$
cannot be distilled.
Since $k$ and $k'$ are arbitrary,
the state $\varrho_N$ is not distillable,
that is, it is BE
since it is inseparable.
Therefore, there exists an $N$-qubit BE state $\varrho_N$
violating the MK inequality if $N\ge 6$.
\end{proof}

As seen in Theorem~\ref{Thm:main2},
for $3\le N\le 5$,
there exists no $N$-qubit BE state that violates the inequality.
Hence we can say that if $3\le N\le 5$ then
violation of the inequality implies distillability.

In conclusion,
we have studied the explicit relation between violation of Bell inequalities and
bipartite distillability of multi-qubit states,
and have shown that if any $N$-qubit state violates the MK inequality then
there exist at least $\lfloor 2^{N-1} - 2^{(N-1)/2}+1\rfloor$ distillable bipartite splits.
Hence, the probability that a randomly chosen bipartite split is distillable
approaches one exponentially with $N$ as $N$ tends to infinity.
We have also shown that an $N$-qubit BE state violates the inequality
if and only if $N\ge 6$.

It has been shown that
while $N$-qubit states in a class of BE states presented in \cite{Dur,Acin}
violate the MK inequality for $N\ge 8$,
the states in the class
violate different forms of Bell inequalities for $N\ge 7$ in Ref.~\cite{KKCO}
and for $N\ge 6$ in Ref.~\cite{SSZ}.
Furthermore, it has been also shown that
there exists a four-qubit BE state
which can maximally violate a certain form of Bell inequality~\cite{AH}.
Therefore, our results could be also improved
by using Bell inequalities different from the MK inequality,
and could be furthermore generalized to multipartite distillability.

S.L. was supported by the Korea Research Foundation Grant funded by the Korean Government
(MOEHRD, Basic Research Promotion Fund) (KRF-2007-331-C00049),
and J.K. was partially supported by the IT R\&D program of MKE/IITA
(2005-Y-001-04, Development of Next Generation Security Technology).


\end{document}